\documentclass[journal]{IEEEtran}

\usepackage{cite}
\ifCLASSINFOpdf
   \usepackage[pdftex]{graphicx}
\else
  \usepackage[dvips]{graphicx}
\fi
\usepackage{amsmath,amssymb,amsthm,mathrsfs,amsfonts,dsfont,stmaryrd}
\usepackage{array}
\usepackage{dsfont}
\usepackage{amsbsy}
\usepackage{epstopdf}
\usepackage{graphicx,color}
\usepackage{enumitem}

\newtheorem{lemma}{Lemma}
\newtheorem{proposition}{Proposition}
\newtheorem{remark}{Remark}

\long\def\symbolfootnote[#1]#2{\begingroup%
\def\thefootnote{\fnsymbol{footnote}}\footnote[#1]{#2}\endgroup}
\newtheorem{theorem}{Theorem}
\newtheorem{definition}{Definition}

\newcommand{\dv}{\mathbf} 
\newcommand{\mc}{\mathcal} 
\newcommand{\mb}{\mathbf} 
\newcommand{\mkv}{-\!\!\!\!\minuso\!\!\!\!-}

 \allowdisplaybreaks

\usepackage{algorithm} 
\usepackage{algpseudocode}

\algnewcommand{\Inputs}[1]{%
  \State \textbf{input:}
   \parbox[t]{.8\linewidth}{\raggedright #1}
}
\algnewcommand{\Initialize}[1]{%
  \State \textbf{initialization}
  \parbox[t]{.95\linewidth}{\raggedright #1}
}

\algnewcommand{\Outputs}[1]{%
  \State \textbf{output:}
   \parbox[t]{.8\linewidth}{\raggedright #1}
}

\begin{document}
\fontencoding{OT1}\fontsize{10}{11}\selectfont

\pagenumbering{gobble}

\title{Distributed Information Bottleneck Method for Discrete and Gaussian Sources}

\author{
I\~naki Estella Aguerri $^{\dagger}$ \qquad \quad Abdellatif Zaidi $^{\dagger}$$^{\ddagger}$ \\   
\small{$^{\dagger}$ Mathematics and Algorithmic Sciences Lab.
France Research Center, Huawei Technologies, Boulogne-Billancourt, 92100, France\\
$^{\ddagger}$ Universit\'e Paris-Est, Champs-sur-Marne, 77454, France\\
\{\tt  inaki.estella@huawei.com, abdellatif.zaidi@u-pem.fr\}}
\vspace{-5mm}
} 


\maketitle
\begin{abstract} 
We study the problem of distributed information bottleneck, in which multiple encoders separately compress their  observations in a manner such that, collectively, the compressed signals preserve as much information as possible about another signal. The model generalizes Tishby's centralized  information bottleneck method to the setting of multiple distributed encoders. We establish single-letter characterizations of the information-rate region of this problem for both i) a class of discrete memoryless sources and ii) memoryless vector Gaussian sources. Furthermore, assuming a sum constraint on rate or complexity, for both models we develop Blahut-Arimoto type iterative algorithms that allow to compute optimal information-rate trade-offs, by iterating over a set of self-consistent equations. 
\end{abstract}

\IEEEpeerreviewmaketitle


\vspace{-4mm}
\section{Introduction}

The information bottleneck (IB) method was introduced by Tishby~\cite{Tishby99theinformation} as an information-theoretic principle for extracting the relevant information that some signal $Y \in \mathcal{Y}$ provides about another one, $X \in \mathcal{X}$, that is of interest.
The approach has found remarkable applications in supervised and unsupervised learning problems such as classification, clustering and prediction, wherein one is interested in extracting the relevant features, i.e., $X$, of the available data $Y$ \cite{Slonim:2000, Slonim20122005}.
Perhaps key to the analysis, and development, of the IB method is its elegant connection with information-theoretic rate-distortion problems. Recent works show that this connection turns out to be useful also for a better understanding of  deep neural networks~\cite{Tishby:2015:ITW}. 
Other connections, that are more intriguing, exist also with seemingly unrelated problems such as hypothesis testing~\cite{Tian:IT:2008} or systems with privacy constraints~\cite{Caire:ISIT:2016}. 

Motivated by applications of the IB method to settings in which the relevant features about $X$ are to be extracted from separately encoded signals, we study the model shown in Figure~\ref{fig:Schm}. Here, $X$ is the signal to be predicted and $(Y_1,\hdots,Y_K)$ are correlated signals that could each be relevant to extract one or more features of $X$. The features could be distinct or redundant. We make the assumption that the signals $(Y_1,\hdots,Y_K)$ are  independent given $X$. This assumption holds in many practical scenarios. For example,  the reader may think of $(Y_1,...,Y_K)$ as being the results of $K$ clinical tests that are performed independently at different clinics and are used to diagnose a disease $X$. A third party (decoder or detector) has to decide without access to the original data. In general, at every encoder $k$ there is a tension among the \textit{complexity} of the encoding, measured by the minimum description length or rate $R_k$ at which the observation is compressed, and the information that the produced description, say $U_k$, provides about the signal $X$. The \textit{relevance} of $(U_1,\hdots,U_K)$ is measured in terms of the information that the descriptions collectively preserve about $X$; and is captured by Shannon's mutual information $I(U_1,\hdots,U_K;X)$. Thus, the performance of the entire system can be evaluated in terms of the tradeoff between the vector $(R_1,\hdots,R_K)$ of minimum description lengths and the mutual information $I(U_1,\hdots,U_K;X)$.

\begin{figure}[t!]
\centering
\includegraphics[width=0.48\textwidth]{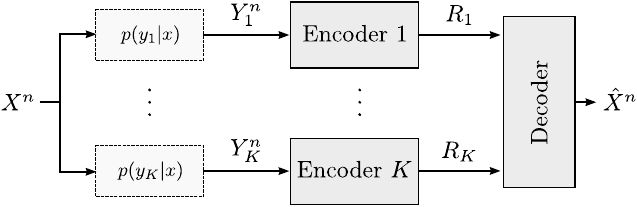}
\vspace{-3mm}
\caption{A model for distributed information bottleneck (D-IB).} 
\vspace{-6mm}
\label{fig:Schm}
\end{figure}

In this paper, we study the aforementioned tradeoff among relevant information and complexity for the model shown in Figure~\ref{fig:Schm}. First, we establish a single-letter characterization of the information-rate region of this model for discrete memoryless sources. In doing so, we exploit its connection with the distributed Chief Executive Officer (CEO) source coding problem under logarithmic-loss distortion measure studied in~\cite{Courtade2014LogLoss}. Next, we extend this result to memoryless vector Gaussian sources. Here, we prove that Gaussian test channels are optimal, thereby generalizing a similar result of~\cite{GlobersonTishby:Techical:Gaussian} and~\cite{journals/jmlr/ChechikGTW05} for the case of a single encoder IB setup.

In a second part of this paper, assuming a sum constraint on rate or complexity,  we develop Blahut-Arimoto~\cite{Blahut:IT:1972} type iterative algorithms that allow to compute optimal tradeoffs between information and rate, for both discrete and vector Gaussian models. We do so through a variational formulation that allows the determination of the set of self-consistent equations satisfied by the stationary solutions. In the  Gaussian case, the algorithm reduces to an appropriate updating of the parameters of  noisy linear projections. Here as well, our algorithms can be seen as generalizations of those developed for the single-encoder IB method, for discrete sources in~\cite{Tishby99theinformation} and for Gaussian sources in~\cite{journals/jmlr/ChechikGTW05}; as well as a generalization of the Blahut-Arimoto algorithm proposed in \cite{UE-AZ17a} for the CEO source coding problem for $K=2$ and discrete sources, to  $K\geq 2$ encoders and for both discrete and Gaussian sources.

\textit{Notation:} Upper case letters  denote random variables, e.g., X;  lower case letters denote realizations of random variables, e.g., $x$; and calligraphic letters denote sets, e.g., $\mathcal{X}$. The cardinality of a set is denoted by $|\mc X|$. For a random variable $X$ with probability mass function (pmf) $P_{X}$, we use $P_{X}(x)=p(x)$, $x\in \mc X$ for short. Boldface upper case letters denote vectors or matrices, e.g., $\dv X$, where context  makes the distinction clear. For an integer $n\in \mathbb{N}$, we denote the set $[1,n]:=\{1,2,\ldots, n\}$.
 We denote by $D_{\mathrm{KL}}(P,Q)$ the Kullback-Leibler divergence between the pmfs $P$ and $Q$.
For a set of integers $\mc K\subseteq \mathds{N}$,  $X_{\mc K}$ denotes the set  $X_{\mc K}=\{X_k:k \in \mc K \}$. 
We denote the covariance of a zero-mean vector $\mathbf{X}$ by $\mathbf{\Sigma}_{\mathbf{x}}:=\mathrm{E}[\mathbf{XX}^H]$;  $\mathbf{\Sigma}_{\mathbf{x},\mathbf{y}}$ is the cross-correlation  $\mathbf{\Sigma}_{\mathbf{x},\mathbf{y}}:= \mathrm{E}[\mathbf{XY}^H]$, and the conditional correlation  of $\mathbf{X}$ given $\mathbf{Y}$ as $\mathbf{\Sigma}_{\mathbf{x}|\mathbf{y}}:= \mathbf{\Sigma}_{\mathbf{x}}-\mathbf{\Sigma}_{\mathbf{x},\mathbf{y}}\mathbf{\Sigma}_{\mathbf{y}}^{-1}\mathbf{\Sigma}_{\mathbf{y},\mathbf{x}}$.

\section{System Model}\label{sec:System}

Consider the discrete memoryless D-IB model shown  in Figure~\ref{fig:Schm}. Let $\{X_{i}, Y_{1,i},\ldots, Y_{K,i}\}_{i=1}^n = (X^n,Y_1^n,\ldots, Y_K^n)$ be a sequence of $n$ independent, identically distributed (i.i.d.) random variables with finite alphabets $\mathcal{X},\mathcal{Y}_k$, $k\in\mathcal{K}:= \{1,\ldots, K\}$ and joint pmf $P_{X,Y_1,\ldots,Y_K}$. Throughout this paper, we make the assumption that the observations at the encoders are independent conditionally on $X$, i.e., 
\begin{align}
Y_{k,i} \mkv X_{i} \mkv Y_{\mathcal{K}/k,i}\quad \text{for } k \in \mc K\: \text{and} \: i \in [1,n]. 
\label{eq:MKChain_pmf}
\end{align}
Encoder $k\in\mathcal{K}$ maps the observed sequence $Y_k^n$ to an index $J_k:= \phi_k(Y_k^n)$, where  $\phi_k: \mathcal{Y}_k^n\rightarrow \mathcal{M}_k$ is a given map and $\mathcal{M}_k:= [1,M_{k}^{(n)}]$. The index $J_k$ is sent error-free to the decoder. 
The decoder collects all indices $J_{\mathcal{K}}:= (J_1,\ldots,J_K)$ and then estimates the source $X^n$ as $\hat{X}^n=g^{(n)}(J_{\mc J})$, where $g^{(n)}: \mathcal{M}_1\times \cdots \times \mathcal{M}_L\rightarrow \mathcal{\hat{X}}^n$ is some decoder map and $\mathcal{\hat{X}}^n$ is the reconstruction alphabet of the source. 

The quality of the reconstruction is measured in terms of the $n$-letter \textit{relevant information} between the unobserved source $X^n$ and its reconstruction at the decoder $\hat{X}^n$, given by
\begin{align}
\Delta^{(n)} := \frac{1}{n} I(X^n;g^{(n)}(\phi_1^{(n)}(Y_1^n),\ldots,\phi_K^{(n)}(Y_K^n) )).
\end{align}


\begin{definition}
A tuple $(\Delta,R_1,\ldots, R_K)$ is said to be achievable for the D-IB model if there exists a blocklength $n$, encoder maps $\phi^{(n)}_k$ for $k\in\mathcal{K}$, and a decoder map $g^{(n)}$, such that
\begin{align}
R_k&\geq \frac{1}{n}\log M_{k}^{(n)}, \;k\in \mc K, \;\;\text{and}\quad
\Delta\leq \frac{1}{n}I(X^n;\hat{X}^n).\label{eq:DeltaConstr}
 \end{align}
where 
$\hat{X}^n=g^{(n)}(\phi_1^{(n)}(Y_1^n),\ldots,\phi_K^{(n)}(Y_K^n) )$.
The information-rate region $\mathcal{R}_{\mathrm{IB}}$ is given by the closure of all achievable rates tuples $(\Delta,R_1,\ldots, R_K)$.
\end{definition}
We are interested in characterizing the  region $\mathcal{R}_{\mathrm{IB}}$.
Due to space limitations, some results are only outlined or provided without proof. We refer to \cite{Estella:IZS:2017:Proofs} for a detailed version.


\section{Information-Rate Region Characterization}

In this section we  characterize the information-rate region $\mathcal{R}_{\mathrm{IB}}$ for a discrete memoryless D-IB model. 
It is well known that the IB problem is essentially a source-coding problem where the distortion measure is of logarithmic loss  type \cite{HT07}. Likewise, the D-IB model of Figure \ref{fig:Schm} is essentially a $K$-encoder CEO source coding problem under logarithmic loss (log-loss) distortion measure.
The log-loss distortion between sequences is defined as
\begin{align}
d_{\mathrm{LL}}(x^n,\hat{x}^n):= -\frac{1}{n}\log\left(\frac{1}{\hat{x}^n(x^n)}\right),
\end{align}
where $\hat{x}^n = s(x^n|j_{\mc K})$ and $s$ is a pmf on $\mc X^n$.

The rate-distortion region of the $K$-encoder CEO source coding problem under log-loss, with $K\geq 2$ which we denote hereafter as $\mc{RD}_{\mathrm{CEO}}$, has been established recently in \cite[Theorem 10]{Courtade2014LogLoss} for the case in which the Markov chain \eqref{eq:MKChain_pmf} holds. 

We first state the following proposition, the proof of which is easy and omitted for brevity.
\begin{proposition}\label{prop:Eqreg}
A tuple $(\Delta,R_1,\ldots, R_K)\in \mc R_{\mathrm{IB}}$ if and only if $(H(X)-\Delta,R_1,\ldots, R_K)\in \mc {RD}_{\mathrm{CEO}}$.
\end{proposition}

Proposition \ref{prop:Eqreg} implies that \cite[Theorem 10]{Courtade2014LogLoss} can be applied to characterize the information-rate region $\mathcal{R}_{\mathrm{IB}}$ as given next.
\begin{theorem}\label{th:MK_C_Main}
In the case in which the Markov chain \eqref{eq:MKChain_pmf} holds, the rate-information region $\mc R_{\mathrm{IB}}$ of the D-IB model is given by the set of all tuples $(\Delta, R_1,\ldots, R_K)$ which satisfy for  $\mathcal{S} \subseteq \mathcal{K}$
\begin{align}\label{eq:MK_C_Main}
\Delta\leq \sum_{k\in \mathcal{S}} [R_k\!-\!I(Y_{k};U_{k}|X,Q)]  + I(X;U_{\mathcal{S}^c}|Q),
\end{align}
for some joint pmf
$p(q) p(x)\prod_{k=1}^K p(y_{k}|x)\prod_{k=1}^{K}p(u_k|y_k,q)$.
\end{theorem} 


\subsection{Memoryless Vector Gaussian D-IB}\label{ssec:Gauss}

Consider now the following memoryless vector Gaussian D-IB problem. In this model, the source vector $\mathbf{X}\in \mathds{C}^{N}$ is Gaussian and has zero mean and  covariance matrix
$\mathbf{\Sigma}_{\mb x}$, i.e., $\mathbf{X}\sim \mc{CN}(\dv 0, \mathbf{\Sigma}_{\mb x})$.
 Encoder $k$, $k\in \mc K$, observes a noisy observation $\mb Y_{k}\in \mathds{C}^{M_k}$, that is given by
\begin{equation}
\mathbf{Y}_k = \mathbf{H}_{k}\mathbf{X}+\mathbf{N}_k,
\label{mimo-gaussian-model}
\end{equation}
where $\mathbf{H}_{k}\in \mathds{C}^{M_k\times N}$ is the channel connecting the source to encoder $k$, and $\mathbf{N}_k\in\mathds{C}^{M_k}$, $k\in \mc K$, is the noise vector at encoder $k$, assumed to be Gaussian, with zero-mean and covariance matrix $\mathbf{\Sigma}_{\mb n_k}$, and independent from all other noises and the source vector $\dv X$.

The studied Gaussian model satisfies the Markov chain~\eqref{eq:MKChain_pmf}; and thus, 
the result of Theorem~\ref{th:MK_C_Main}, which can be extended to continuous sources using standard techniques, characterizes the information-rate region of this model. The following theorem characterizes  $\mc R_{\mathrm{IB}}$ for the vector Gaussian model, shows that the optimal test channels $P_{U_k|Y_k}$, $k\in \mc K$, are Gaussian and that there is not need for time-sharing, i.e., $Q=\emptyset$.

\begin{theorem}
\label{th:GaussSumCap}
If $(\dv X,\dv Y_1,\ldots, \dv Y_K)$ are jointly Gaussian as in \eqref{mimo-gaussian-model}, the information-rate region $\mc R$ is given by the set of all  tuples $(\Delta, R_1,\ldots,R_L)$ satisfying that for all $\mathcal{S} \subseteq \mathcal{K}$
\begin{align}
\Delta \leq
 \sum_{k\in \mathcal{S}}\left[R_k+\log|\mathbf{I}-\mathbf{B}_k|\right] 
 + \log 
\left|\sum_{k\in\mathcal{S}^{c}}\mathbf{\bar{H}}_{k}^{H}
\mathbf{B}_{k}
\mathbf{\bar{H}}_{k}+\mathbf{I}\right|\nonumber
,
\end{align}
for some $\mathbf{0}\preceq \mathbf{B}_k\preceq \mathbf{I}$ and where $\mathbf{\bar{H}}_{k} = \mathbf{\Sigma}_{\mb n_k}^{-1/2} \mathbf{H}_{k} \mathbf{\Sigma}_{\mb x}^{1/2}$. In addition, the information-rate tuples in $\mc R_{\mathrm{IB}}^{\mathrm{G}}$ are achievable with $Q= \emptyset$ and $p^{*}(\dv u_k|\dv y_k,q) = \mc{CN}(\dv y_k, \dv \mathbf{\Sigma}_{\mb n_k}^{1/2}(\dv B_k-\dv I)\mathbf{\Sigma}_{\mb n_k}^{1/2}  ) $.
\end{theorem}
\begin{proof}
An outline of the proof is given in Appendix~\ref{app:GaussSumCap}.
\end{proof}

%

\section{Computation of the Information Rate Region under Sum-Rate Constraint}

In this section, we describe an iterative Blahut-Arimoto (BA)-type algorithm to compute the pmfs $P_{U_k|Y_k}$, $k\in \mc K$,  that maximize information $\Delta$  under sum-rate constraint, i.e., $R_{\mathrm{sum}} := \sum_{k=1}^KR_k$, for tuples $(\Delta,R_1,\ldots, R_K)$ in $\mc R_{\mathrm{IB}}$. From Theorem~\ref{th:MK_C_Main} we have:
\begin{align}
\mc R_{\mathrm{sum}}:= \textrm{convex-hull}\{(\Delta,R_{\mathrm{sum}}):\Delta\leq \Delta_{\mathrm{sum}}( R_{\mathrm{sum}} )\},
\end{align}
where we define the information-rate function
\begin{align}
\Delta_{\mathrm{sum}}(R) :=\! \max_{\dv P}  \min\left\{I(X; U_{\mc K}),R - \sum_{k=1}^KI(Y_k;U_k|X)\right\},\nonumber
\end{align}
and where the optimization is over the set of $K$ conditional pmfs $P_{U_k|Y_k}$, $k\in \mc K$, which, for short, we define as
\begin{align}
\dv P := \{P_{U_1|Y_1},\ldots, P_{U_K|Y_K}\}.
\end{align}

Next proposition  provides a  characterization of the pairs $(\Delta,R_{\mathrm{sum}})\in \mc R_{\mathrm{sum}}$ in terms of a parameter $s\geq 0$.
\begin{proposition}\label{prop:param}
Each tuple $(\Delta,R_{\mathrm{sum}})$ on the information rate curve  $\Delta = \Delta_{\mathrm{sum}}(R_{\mathrm{sum}})$, can be obtained for some 
$s \geq 0$, as $(\Delta_{s}, R_{s})$, parametrically defined by 
\begin{align}
&(1+s)\Delta_{s} = (1+sK)H(X) +  s R_{s} - \min_{\dv P}F_{s}(\dv P),\label{eq:Dparam}\\
&R_{s} = I(Y_{\mc K};U_{\mc K}^*) + \sum_{k=1}^K [I(Y_k;U_k^*) - I(X;U_k^*)],\label{eq:R1param}
\end{align}
where $\dv P^*$ are the pmfs yielding the minimum in \eqref{eq:Dparam} and 
\begin{align}
F_{s}(\dv P) :=  H(X|U_{\mc K}) + s \sum_{k=1}^K [I(Y_k;U_k) + H(X|U_k)].
\end{align}

\end{proposition}
%
%

\begin{proof}
The proof of Proposition~\ref{prop:param} follows along the lines of \cite[Theorem 2]{UE-AZ17a} and is omitted for brevity.  Note that the rate expression in Theorem~\ref{th:MK_C_Main} is different to that in \cite{UE-AZ17a}.
\end{proof}

From  Proposition~\ref{prop:param}, the information-rate function can be computed by solving \eqref{eq:Dparam} and evaluating \eqref{eq:R1param} for all $s\geq 0$. Inspired by the standard Blahut-Arimoto (BA) method\cite{Blahut:IT:1972}, and following similar steps as for the BA-type algorithm proposed in \cite{UE-AZ17a} for the CEO problem with $K=2$ encoders, we show that problem \eqref{eq:Dparam} can be solved with an alternate optimization procedure, with respect to $\dv{P}$ and some appropriate auxiliary pmfs $Q_{X|U_k}$, $k\in \mc K$ and $Q_{X|U_{1},\dots, U_K}$, denoted for short as
\begin{align}
\dv Q := \{Q_{X|U_1},\ldots, Q_{X|U_K},Q_{X|U_1,\ldots, U_K}\}.
\end{align}
To this end, we define the  function $\bar{F}_s(\cdot)$ and write \eqref{eq:Dparam} as a minimization over the pmfs $\dv P$ and pmfs $\dv Q$, where 
\vspace{-1mm}
\begin{align}
\bar{F}_s(\dv P, \dv Q) := & \; s \sum_{k=1}^K I(Y_k;U_k)- s \sum_{k=1}^K \mathrm{E}_{X,U_k}[\log  q(X|U_k)]\nonumber\\
 &- \mathrm{E}_{X,U_{\mc{K}}}[\log q(X|U_1,\ldots, U_K)].\label{eq:FunctionPQ}
\end{align}
\begin{lemma}\label{lem:VarEq} We have
\vspace{-1mm}
\begin{align}
F^*:=\min_{\dv P}F_{s}(\dv P) = \min_{\dv P}\min_{\dv Q}\bar{F}_{s}(\dv P,\dv Q).\label{eq:VarEq}
\end{align}
\end{lemma}
\vspace{-2mm}
Algorithm \ref{algo:BA_DMC} describes the steps to  successively minimize $\bar{F}_{s}(\dv P,\dv Q)$ by optimizing a convex problem over $\dv P$ and over $\dv Q$ at each iteration.
 The proof of Lemma \ref{lem:VarEq} and the steps of the proposed algorithm are justified with the following lemmas, whose proofs are along the lines of Lemma 1, Lemma 2, Lemma 3 in \cite{UE-AZ17a}, and are omitted due to space limitations.
\begin{lemma}\label{lem:convex}
 $\bar{F}_s(\dv P, \dv Q)$ is convex in $\dv P$ and convex in $\dv Q$.
 \end{lemma}

\begin{lemma}\label{lemma:QUpdate}
For fixed pmfs $\dv P$, $\bar{F}_s(\dv P, \dv Q) \geq F_s(\dv P)$ for all pmfs $\dv Q$, and
 there exists a unique $\dv Q$ that achieves the minimum $\min_{\dv Q}\bar{F}_s(\dv P, \dv Q) = F_s(\dv P)$, given by
\begin{align}
Q^*_{X|U_k} &= P_{X|U_k},\quad k\in \mc K, \label{eq:Qstark}\\
Q^*_{X|U_1,\ldots,U_k} &= P_{X|U_1,\ldots, U_K}, \label{eq:Qstarall}
\end{align}
where $P_{X|U_k}$ and $P_{X|U_1,\ldots, U_K}$ are computed from $\dv P$.
\end{lemma}
\begin{lemma}
For fixed $\dv Q$, there exists a  $\dv P$ that achieves the minimum $\min_{\dv P}\bar{F}_s(\dv P, \dv Q)$, where $P_{U_k|Y_k}$ is given by
\begin{align}
p^*(u_k|y_k) = p(u_k)\frac{\exp\left(-\psi_s(u_k,y_k)\right)}{\sum_{y_k\in \mc Y_k} \exp(-\psi_s(u_k,y_k))},\label{eq:P_update}
\end{align}
for $u_k\in \mc U_k$ and $y_k\in \mc Y_k$, $k\in \mc K$,  and where we define 
\begin{align}\label{eq:P_update_psi}
\psi_s(u_k,y_k):= & D_{\mathrm{KL}}(P_{X|y_k}||Q_{X|u_k})\\
&+\frac{1}{s}
\mathrm{E}_{U_{\mc K\setminus k}|y_k}[D_{\mathrm{KL}}(P_{X|U_{\mc K\setminus k},y_k}||Q_{X|U_{\mc K\setminus k},u_k}))]\nonumber.
\end{align}
\end{lemma}

\begin{algorithm}
\caption{BA-type algorithm for the Discrete D-IB}\label{algo:BA_DMC}
\begin{algorithmic}[1]
\smallskip
\Inputs{pmf $P_{X,Y_1,\ldots,Y_k}$, parameter $s\geq 0$.}
\Outputs{optimal $P^*_{U_k|Y_k}$, pair $(\Delta_s,R_s)$.}
\Initialize{Set $t=0$ and set $\dv P^{(0)}$ with $p(u_k|y_k)= \frac{1}{|\mc U_k|}$ \\ for $u_k\in \mc U_k$, $y_k\in \mc Y_k$, $k=1,\ldots, K$.}
\Repeat 
\State Compute $\dv Q^{(t+1)}$ as \eqref{eq:Qstark} and \eqref{eq:Qstarall} from $\dv P^{(t)}$.
\State Compute $\dv P^{(t+1)}$ as \eqref{eq:P_update} from $\dv Q^{(t+1)}$ and $\dv P^{(t)}$.
\State $t \leftarrow t+1$.
\Until{convergence.}
\end{algorithmic}
\end{algorithm}

 Algorithm \ref{algo:BA_DMC} essentially falls in the Successive Upper-Bound Minimization (SUM) framework \cite{Razaviyayn:SIAM:UnifiedConvergence}  in which $\bar{F}_s(\dv P, \dv Q)$ acts as a globally tight upper bound on $F_s(\dv P)$. Algorithm \ref{algo:BA_DMC} provides a sequence $\dv P^{(t)}$ for each iteration $t$, which  converges to a stationary point of the optimization problem~\eqref{eq:VarEq}. 
\vspace{-1mm}
\begin{proposition}
Every limit point of the sequence $\dv P^{(t)}$ generated by Algorithm~\ref{algo:BA_DMC} converges to a stationary point of~\eqref{eq:VarEq}.
\end{proposition}
\vspace{-3mm}
\begin{proof}
Let $\dv Q^*(\dv P):= \arg\min_{\dv Q}\bar{F}_s(\dv P, \dv Q)$. From Lemma~\ref{lemma:QUpdate},  $\bar{F}_s(\dv P, \dv Q^*(\dv P'))\geq \bar{F}_s(\dv P, \dv Q^*(\dv P)) = F_s(\dv P) $ for $\dv P'\neq \dv P$.  It follows that $F_s(\dv P)$ and $\bar{F}_s(\dv P, \dv Q^*(\dv P'))$ satisfy \cite[Proposition 1]{Razaviyayn:SIAM:UnifiedConvergence}  and thus $\bar{F}_s(\dv P, \dv Q^*(\dv P'))$  satisfies A1-A4 in \cite{Razaviyayn:SIAM:UnifiedConvergence}. Convergence to a stationary point of~\eqref{eq:VarEq} follows from \cite[Theorem 1]{Razaviyayn:SIAM:UnifiedConvergence}.
\end{proof}
\vspace{-3mm}

\begin{remark}
The resulting set of self consistent equations  \eqref{eq:Qstark}, \eqref{eq:Qstarall} and \eqref{eq:P_update_psi} satisfied by any stationary point of the D-IB problem, remind that of the original IB problem~\cite{GlobersonTishby:Techical:Gaussian}.
 Note the additional divergence term in  \eqref{eq:P_update_psi} for encoder $k$ averaged over the descriptions at the other $\mc K\setminus k$ encoders.
\end{remark}
\vspace{-4mm}

\section{Computation of the Information Rate Region for the Vector Gaussian D-IB}

Computing the maximum information under sum-rate constraint from Theorem~\ref{th:GaussSumCap} is a convex optimization problem on $\dv B_k$, which can be efficiently solved with generic tools. 
Alternatively, next we extend Algorithm~\ref{algo:BA_DMC} for Gaussian  sources.

For finite alphabet sources the updates of  $\dv Q^{(t+1)}$ and $\dv P^{(t+1)}$ in Algorithm~\ref{algo:BA_DMC} are simple, but become unfeasible for continuous alphabet sources. We leverage on the optimality of Gaussian descriptions, shown in  Theorem~\ref{th:GaussSumCap}, to restrict the optimization of $\dv P$ to Gaussian distributions, which are easily represented by a finite set of parameters, namely its mean and covariance. We show that if $\dv P^{(t)}$ are Gaussian pmfs, then $\dv{P}^{(t+1)}$ are also  Gaussian pmfs, which can be computed with an efficient update algorithm of its representing parameters. In particular,  if at time $t$, the $k$-th pmf $P_{\dv U_k|\dv Y_k}^{(t)}$ is given by
\begin{align}
\dv U_k^{t} = \dv A_{k}^{t}\dv Y_k +\dv Z_{k}^{t},\label{eq:testChan}
\end{align}
where $\dv Z_{k}^{t}\sim\mc{CN}(\dv 0,\dv \Sigma_{\dv z_{k}^{t}})$; we show that for $\dv P^{(t+1)}$  updated as in \eqref{eq:P_update},  $P_{\dv U_{k}|\dv Y_k}^{(t+1)}$ corresponds to $\dv U_{k}^{t+1} = \dv A_{k}^{t+1}\dv Y_{k}+\dv Z_{k}^{t+1}$,
where $\dv Z_{k}^{t+1}\!\sim\mc{CN}(\dv 0, \dv\Sigma_{\dv z_{k}^{t+1}})$  and  $\dv A_{k}^{t+1}, \dv\Sigma_{\dv z_{k}^{t+1}}$ are updated as
\begin{align}
\dv \Sigma_{\dv z_k^{t+1}} =&\left(\left(1+\frac{1}{s}\right)\dv \Sigma_{\dv u_k^t|\dv x}^{-1}  - \frac{1}{s} \dv \Sigma_{\dv u_k^t|\dv u_{\mc K\setminus k}^t}^{-1}\right)^{-1},\label{eq:SigmaUpdate}\\
\dv A_{k}^{t+1} =&\dv \Sigma_{\dv z_k^{t+1}}^{-1} \left(\left(1+\frac{1}{s}\right)\dv \Sigma_{\dv u_k^t|\dv x}^{-1}\dv  A_{k}^t(\dv I - \dv \Sigma_{\dv y_k|\dv x}\dv \Sigma_{\dv y_k}^{-1})\right.\nonumber\\
&\left.-\frac{1}{s}\dv \Sigma_{\dv u_k^t|\dv u_{\mc K\setminus k}^t}^{-1}\dv  A_{k}^t(\dv I - \dv \Sigma_{\dv y_k|\dv u_{\mc K\setminus k}^t}\dv \Sigma_{\dv y_k}^{-1})\right).\label{eq:AUpdate}
\end{align}
The detailed update procedure is given in Algorithm~\ref{algo:BA_Gauss}.
\newpage
\begin{remark}
Algorithm~\ref{algo:BA_Gauss} generalizes the iterative algorithm  for single encoder Gaussian D-IB in~\cite{journals/jmlr/ChechikGTW05}  to the Gaussian D-IB  with $K$ encoders and sum-rate constraint. Similarly to the solution in~\cite{journals/jmlr/ChechikGTW05}, the optimal description at each encoder is given by a noisy linear projection of the observation, whose dimensionality is determined by the parameter $s$ and the second order moments between the observed data and the source of interest, as well as a term depending on the  observed data with respect to the descriptions at the other encoders.
\end{remark}

\subsection{Derivation of Algorithm~\ref{algo:BA_Gauss}}

In this section, we derive the update rules in Algorithm~\ref{algo:BA_Gauss} and show that the Gaussian distribution is invariant to the update rules in Algorithm~\ref{algo:BA_DMC}, in line with Theorem~\ref{th:GaussSumCap}.

First, we recall that if $(\dv X_1,\dv X_2)$ are jointly Gaussian, then
\begin{align}
P_{\dv X_2|\dv X_1 = \dv x_1} = \mc{CN}(\boldsymbol\mu_{\dv x_2|\dv x_1},\dv\Sigma_{\dv x_2|\dv x_1}),
\end{align}
where
$\boldsymbol\mu_{\dv x_2|\dv x_1}:= \dv K_{\dv x_2|\dv x_1}\dv x_1$, with $\dv K_{\dv x_2|\dv x_1}:=\dv\Sigma_{\dv x_2,\dv x_1}\dv \Sigma_{\dv x_1}^{-1}$ .

Then, for $\dv Q^{(t+1)}$ computed as in \eqref{eq:Qstark} and \eqref{eq:Qstarall} from $\dv P^{(t)}$, which is a set of Gaussian distributions,  we have
\begin{align}
 Q^{(t+1)}_{\dv X|\dv{u}_k} = \mc {CN}(\boldsymbol \mu_{\dv x|\dv u_{k}^t} , \dv \Sigma_{\dv x|\dv u_{k}^t}),
Q^{(t+1)}_{\dv X|\dv{u}_{\mc K}} = \mc {CN}(\boldsymbol \mu_{\dv x|\dv u_{\mc K}^t}, \dv \Sigma_{\dv x|\dv u_{\mc K}^t})\nonumber.
\end{align}

Next, we look at the update  $\dv P^{(t+1)}$ as in \eqref{eq:P_update} from given $\dv Q^{(t+1)}$. First, we have that $p(\dv u_{k}^t)$  is the marginal of $\dv U_{k}^t$, given by $\dv U_{k}^t\sim \mc{CN}(\dv 0,\dv \Sigma_{\dv u_{k}^t} )$  where $\dv \Sigma_{\dv u_{k}^t} = \dv A_{k}^t\dv \Sigma_{\dv y_k} \dv A_{k}^{t,H} + \dv\Sigma_{\dv z_{k}^t}$. 

Then, to compute $\psi_s(\dv u_k^t,\dv y_k)$, first, we note that 
\begin{align}
E&_{ U_{\mc K\setminus k}|y_k }[D_{\mathrm{KL}}(P_{ X| U_{\mc K\setminus k},y_k}||Q_{X| U_{\mc K\setminus k},u_k})]\label{eq:DistEq}\\
=&D_{\mathrm{KL}}(P_{ X, U_{\mc K\setminus k}|y_k}||Q_{ X,U_{\mc K\setminus k}|u_k})\!-\!D_{\mathrm{KL}}(P_{ U_{\mc K\setminus k}|y_k}||Q_{ U_{\mc K\setminus k}|u_k})\nonumber,
\end{align}
and that for two generic multivariate Gaussian distributions $P_1\sim\mc{CN}(\boldsymbol \mu_1,\dv \Sigma_1)$ and  $P_2\sim\mc{CN}(\boldsymbol \mu_2,\dv \Sigma_2)$ in $\mathds{C}^N$,
\begin{align}
D_{\mathrm{KL}}(P_1,P_2) =& (\boldsymbol\mu_1-\boldsymbol\mu_2)^H\dv\Sigma_{2}^{-1}(\boldsymbol\mu_1-\boldsymbol\mu_2)\nonumber\\
&+\log |\dv\Sigma_2\dv\Sigma_1^{-1}| - N +\mathrm{tr}\{\dv \Sigma_2^{-1}\dv\Sigma_1\}.\label{eq:DistGaussi}
\end{align}

Applying \eqref{eq:DistEq} and \eqref{eq:DistGaussi} in \eqref{eq:P_update_psi} and noting that all involved distributions are Gaussian, it follows that $\psi_s(\dv u_k^t,\dv y_k)$ is a quadratic form. Then, since $p(\dv u_k^t)$ is Gaussian, the product $\log (p(\dv u_k^t)\exp(-\psi_s(\dv u_k^t,\dv y_k)))$ is also a quadratic form, and identifying constant, first and second order terms, we can write
\begin{align}
\log p^{(t+1)}(\dv u_k|\dv y_k) =& Z(\dv y_k)+ (\dv u_k-\boldsymbol\mu_{\dv u_{k}^{t+1}| \dv y_k})^{H}\dv \Sigma_{\dv z_k^{t+1}}^{-1}\nonumber\\
&\cdot (\dv u_k-\boldsymbol\mu_{\dv u_{k}^{t+1}|\dv y_k}),
\end{align}
where $ Z(\dv y_k)$ is a normalization term independent of $\dv u_k$, and 
\begin{align}
\dv \Sigma_{\dv z_k^{t+1}}^{-1} =& \dv \Sigma_{\dv u_k^t}^{-1} + \dv K_{\dv x|\dv u_k^t}^H \dv \Sigma_{\dv x| \dv u_k}^{-1}\dv K_{\dv x|\dv u_k^t}\nonumber\\
&
+\frac{1}{s}\dv K_{\dv x\dv u_{\mc K\setminus k}^t|\dv u_k^t}^H \dv \Sigma_{\dv x\dv u_{\mc K\setminus k}^t| \dv u_k}^{-1}\dv K_{\dv x\dv u_{\mc K\setminus k}^t|\dv u_k^t}\nonumber  \\
&
- \frac{1}{s} \dv K_{\dv u_{\mc K\setminus k}^t|\dv u_k^t}^H \dv \Sigma_{\dv u_{\mc K\setminus k}^t| \dv u_k}^{-1}\dv K_{\dv u_{\mc K\setminus k}^t|\dv u_k^t}\label{eq:SecondOrder}, \\
\boldsymbol\mu_{\dv u_{k}^{t+1}}=&\dv\Sigma_{\dv z_{k}^{t+1}}\left(  \dv K_{\dv x|\dv u_k^t}^H\dv \Sigma_{\dv x|\dv u_k^t}^{-1}\boldsymbol\mu_{\dv x|\dv y_k}\right.\nonumber\\
&+\frac{1}{s}\dv K_{\dv x,\dv u_{\mc K\setminus k}^t|\dv u_k^t} \dv \Sigma_{\dv x,\dv u_{\mc K\setminus k}^t|\dv u_k^t}^{-1}\boldsymbol\mu_{\dv x,\dv u_{\mc K\setminus k}^t|\dv y_k}\nonumber\\
&\left.-\frac{1}{s} \dv K_{\dv u_{\mc K\setminus k}^t|\dv u_k^t}\dv \Sigma_{\dv u_{\mc K\setminus k}^t|\dv u_k^t}^{-1}\boldsymbol\mu_{\dv u_{\mc K\setminus k}^t|\dv y_k}\label{eq:FirstOrder}\right). 
\end{align}
This shows that $p^{(t+1)}(\dv u_k|\dv y_k)$ is a Gaussian distribution and that $\dv U_{k}^{t+1}$ is distributed as $\dv U_{k}^{t+1}\sim \mc {CN}(\boldsymbol\mu_{\dv u_{k}^{t+1}},\dv\Sigma_{\dv z_{k}^{t+1}})$.

\begin{algorithm}
\caption{BA-type algorithm for the Gaussin Vector D-IB}\label{algo:BA_Gauss}
\begin{algorithmic}[1]
\smallskip
\Inputs{covariance ${\dv \Sigma}_{\dv x, \dv y_1,\ldots,\dv y_k}$, parameter $s\geq 0$.}
\Outputs{optimal pairs $(\dv A_k^{*},\dv \Sigma_{\dv z_k^{*}} )$, $k=1,\ldots, K$.}
\Initialize{Set randomly $\dv A_k^{0}$ and $\dv \Sigma_{\dv z_k^{0}} \succeq 0$, $k\in \mc K$.}
\Repeat 
\State Compute $\dv \Sigma_{\dv y_k|\dv u_{\mc K\setminus k}^t}$ and update for $k\in \mc K$
\begin{align}
\dv \Sigma_{\dv u_k^t|\dv x} &= \dv A^{t}_k \dv \Sigma_{\dv y_k|\dv x} \dv A^{t,H}_k + \dv \Sigma_{\dv z_k^t}\label{eq:Cov_ux}\\
\dv \Sigma_{\dv u_k^t|\dv u_{\mc K\setminus k}^t} &= \dv A^{t}_k \dv \Sigma_{\dv y_k|\dv u_{\mc K\setminus k}^t} \dv A^{t,H}_k + \dv \Sigma_{\dv z_k^t},\label{eq:Cov_uus}
\end{align}

\State Compute $\dv \Sigma_{\dv z_k^{t+1}}$ as in \eqref{eq:SigmaUpdate} for $k\in \mc K$.
\State Compute $\dv A_k^{t+1}$ as \eqref{eq:AUpdate}, $k\in \mc K$.
\State $t \leftarrow t+1$.
\Until{convergence.}
\end{algorithmic}
\end{algorithm}

Next, we simplify \eqref{eq:SecondOrder} and \eqref{eq:FirstOrder} to obtain the update rules \eqref{eq:SigmaUpdate} and \eqref{eq:AUpdate}. From the matrix inversion lemma, similarly to \cite{journals/jmlr/ChechikGTW05}, for $(\dv X_1,\dv X_2)$ jointly Gaussian  we have 
\begin{align}\label{eq:InvLemma1}
\dv \Sigma_{\dv x_2|\dv x_1}^{-1} = \dv\Sigma_{\dv x_2}^{-1} + \dv K_{\dv x_1|\dv x_2}^{H}\dv\Sigma_{\dv x_1|\dv x_2}^{-1}\dv K_{\dv x_1|\dv x_2}.
\end{align}
Applying \eqref{eq:InvLemma1}, in \eqref{eq:SecondOrder} we have
\begin{align}
\dv \Sigma_{\dv z_k^{t+1}}^{-1} 
&=\dv \Sigma_{\dv u_k^t|\dv x}^{-1} +\frac{1}{s}\dv \Sigma_{\dv u_k^t|\dv x \dv u_{\mc K\setminus k}^t}^{-1} - \frac{1}{s} \dv \Sigma_{\dv u_k^t|\dv u_{\mc K\setminus k}^t}^{-1}   
\label{eq:SecondOrder_subLemma},\\ 
&=\left(1+\frac{1}{s}\right)\dv \Sigma_{\dv u_k^t|\dv x}^{-1}  - \frac{1}{s} \dv \Sigma_{\dv u_k^t|\dv u_{\mc K\setminus k}^t}^{-1},\label{eq:SecondOrder_subLemma_MK}
\end{align}
where  \eqref{eq:SecondOrder_subLemma_MK} is due to the Markov chain $\dv U_k\mkv \dv X\mkv \dv U_{\mc K\setminus k }$.

Then, also from the matrix inversion lemma,  we have for jointly Gaussian $(\dv X_1,\dv X_2)$, 
\begin{align}\label{eq:InvLemma2}
\dv \Sigma_{\dv x_2|\dv x_1}^{-1} \dv \Sigma_{\dv x_1,\dv x_2} \dv\Sigma_{\dv x_1}^{-1} =\dv \Sigma_{\dv x_2}^{-1} \dv \Sigma_{\dv x_1,\dv x_2} \dv\Sigma_{\dv x_1|\dv x_2}^{-1}.
\end{align}
Applying \eqref{eq:InvLemma2} in \eqref{eq:FirstOrder}, for the first term, we have
\begin{align}
 \dv K_{\dv x|\dv u_k^t}^H\dv \Sigma_{\dv x|\dv u_k^t}^{-1}\boldsymbol\mu_{\dv x|\dv y_k}\!=\!&
\dv \Sigma_{\dv u_k^t|\dv x}^{-1}\dv\Sigma_{\dv x,\dv u_k^t}\dv \Sigma_{\dv x}^{-1}\boldsymbol\mu_{\dv x|\dv y_k}\\
=&
\dv \Sigma_{\dv u_k^t|\dv x}^{-1}\dv  A_{k}^{t}\dv \Sigma_{\dv y_k,\dv x}\dv \Sigma_{\dv x}^{-1}\dv\Sigma_{\dv x,\dv y_k}\dv \Sigma_{\dv y_k}^{-1}\dv y_k\nonumber\\
=&
\dv \Sigma_{\dv u_k^t|\dv x}^{-1}\dv  A_{k}^{t}(\dv I - \dv \Sigma_{\dv y_k|\dv x}\dv \Sigma_{\dv y_k}^{-1}) \dv y_k, \label{eq:eq:FirstOrder_eq_Invlemma1}
\end{align}
where $\dv\Sigma_{\dv x,\dv u_k^t}=\dv  A_{k}^{t}\dv \Sigma_{\dv y_k,\dv x}$; and \eqref{eq:eq:FirstOrder_eq_Invlemma1} is due to the definition of $\dv\Sigma_{\dv y_k |\dv x}$.
Similarly, for the second term, we have
\begin{align}
\dv K_{\dv x\dv u_{\mc K\setminus k}^t|\dv u_k^t}& \dv \Sigma_{\dv x\dv u_{\mc K\setminus k}^t|\dv u_k^t}^{-1}\boldsymbol\mu_{\dv x,\dv u_{\mc K\setminus k}^t|\dv y_k}\nonumber\\
=&
\dv \Sigma_{\dv u_k^t|\dv x\dv u_{\mc K\setminus k}^t}^{-1}\dv  A_{k}^{t}(\dv I - \dv \Sigma_{\dv y_k|\dv x\dv u_{\mc K\setminus k}^t}\dv \Sigma_{\dv y_k}^{-1}) \dv y_k,\\
=&
\dv \Sigma_{\dv u_k^t|\dv x}^{-1}\dv  A_{k}^{t}(\dv I - \dv \Sigma_{\dv y_k|\dv x}\dv \Sigma_{\dv y_k}^{-1}) \dv y_k, \label{eq:eq:FirstOrder_eq_MK_chain}
\end{align}
where we use $\dv\Sigma_{\dv u_k^t,\dv x\dv u_{\mc K\setminus k}^t}=\dv  A_{k}^t\dv \Sigma_{\dv y_k,\dv x\dv u_{\mc K\setminus k}^t}$; and \eqref{eq:eq:FirstOrder_eq_MK_chain} is due to the Markov chain $\dv U_k\mkv \dv X\mkv \dv U_{\mc K\setminus k }$.
For the third term,
\begin{align}
\dv K_{\dv u_{\mc K\setminus k}^t|\dv u_k^t}& \dv \Sigma_{\dv u_{\mc K\setminus k}^t|\dv u_k^t}^{-1}\boldsymbol\mu_{\dv u_{\mc K\setminus k}^t|\dv y_k}\nonumber\\
=&
\dv \Sigma_{\dv u_k^t|\dv u_{\mc K\setminus k}^t}^{-1}\dv  A_{k}^{t}(\dv I - \dv \Sigma_{\dv y_k|\dv u_{\mc K\setminus k}^t}\dv \Sigma_{\dv y_k}^{-1}) \dv y_k.
\end{align}

Equation \eqref{eq:AUpdate} follows by noting that $\boldsymbol\mu_{\dv u_{k}^{t+1}} = \dv A_{k}^{t+1}\dv y_k$, and that from \eqref{eq:FirstOrder} $\dv A_{k}^{t+1}$ is given as in \eqref{eq:AUpdate}. 

Finally, we note that due to \eqref{eq:testChan}, $\dv\Sigma_{\dv u_k|\dv x}$ and $\dv\Sigma_{\dv u_k|\dv u_{\mc K\setminus k}^t}$ are given as in \eqref{eq:Cov_ux} and \eqref{eq:Cov_uus}, where $\dv\Sigma_{\dv y_k|\dv x}=\dv\Sigma_{\dv n_k}$ and $\dv\Sigma_{\dv y_k|\dv u_{\mc K\setminus k}^t}$ can be computed from its definition.

\section{Numerical Results}

In this section, we consider the numerical evaluation of Algorithm~\ref{algo:BA_Gauss}, and compare the resulting relevant information to two upper bounds on the performance for the D-IB: i) the information-rate pairs achievable under centralized IB encoding,  i.e., if $(Y_1,\ldots, Y_K)$ are encoded jointly at a rate equal to the total rate $R_{\mathrm{sum}}= R_1+\cdots+R_{K}$, characterized in \cite{journals/jmlr/ChechikGTW05}; ii) the information-rate pairs achievable under centralized IB encoding when $R_{\mathrm{sum}}\rightarrow \infty$, i.e., $\Delta = I(X;Y_1,\ldots, Y_K)$.

Figure \ref{fig:Figure1} shows the resulting $(\Delta,R_{\mathrm{sum}})$ tuples for a Gaussian vector model with $K=2$ encoders, source dimension $N=4$, and observations dimension $M_1 = M_2 = 2$ for different values of $s$ calculated as in Proposition~\ref{prop:param} using Algorithm~\ref{algo:BA_Gauss}, and its upper convex envelope. As it can be seen, the distributed IB encoding of sources performs close to the Tishby's centralized IB method, particularly for low $R_{\mathrm{sum}}$ values. Note the discontinuity in the curve caused by a dimensionality change in the projections at the encoders. 

\vspace{-2mm}

\appendices
\section{Proof Outline of Theorem~
\ref{th:GaussSumCap}
}\label{app:GaussSumCap}
Let $(\mathbf{X},\mathbf{U})$ be two complex random vectors. The conditional Fischer information is defined as
$\mathbf{J}(\mathbf{X}|\mathbf{U}) := \mathrm{E}[\nabla \log p(\mathbf{X}|\mathbf{U})\nabla\log p(\mathbf{X}|\mathbf{U})^H]$,
and the MMSE is given by 
$\mathrm{mmse}(\mb X|\mb U) := \mathrm{E}[(\mb X-\mathrm{E}[\mb X|\mb U])(\mb X-\mathrm{E}[\mb X|\mb U])^H]$. Then~\cite{Ekrem:IT:2014:OuterBoundCEO}
\begin{align}
\hspace{-3.1mm}\log|(\pi e) \mathbf{J}^{-1}(\mathbf{X}|\mathbf{U})|\!\leq\! h(\mathbf{X}|\mathbf{U})\!\leq\! \log|(\pi e) \mathrm{mmse}(\mb X|\mb U)|.\hspace{-2mm}\label{lem:FI_Ineq}
\end{align}


We outer bound the information-rate region in Theorem~\ref{th:MK_C_Main} for $(\dv X,\dv Y_{\mc K})$ as in \eqref{mimo-gaussian-model}. 
For $ q\in \mc{Q}$ and fixed $\prod_{k=1}^{K}p(\mathbf{u}_k|\mathbf{y}_k,q)$, choose $\mathbf{B}_{k,q}$, $k\in \mc K$ satisfying $\mathbf{0}\preceq\mathbf{B}_{k,q}\preceq\mathbf{\Sigma}_{\mb n_k}^{-1}$ such that 
\begin{align}
\mathrm{mmse}(\mathbf{Y}_k|\mathbf{X}, \mathbf{U}_{k,q},q) = \mathbf{\Sigma}_{\mb n_k}-\mathbf{\Sigma}_{\mb n_k}\mathbf{B}_{k,q}\mathbf{\Sigma}_{\mb n_k}.\label{eq:covB}
\end{align}
Such $\mathbf{B}_{k,q}$ always exists since $ \mathbf{0}\preceq\mathrm{mmse}(\mathbf{Y}_k|\mathbf{X},\mathbf{U}_{k,q},q)\preceq \mathbf{\Sigma}_{\mb n_k}^{-1}$, for all $q\in \mc Q$, and $k\in \mc K$. We have from \eqref{eq:MK_C_Main},
\begin{align}
I(\mathbf{Y}_k;\mathbf{U}_k|\mathbf{X},q)
& \geq \log|\boldsymbol\Sigma_{\mb n_k}| -\log|\mathrm{mmse}(\mathbf{Y}_k|\mathbf{X},\mathbf{U}_{k,q},q) |\nonumber\\
&= - \log|\dv I-\mathbf{\Sigma}_{\mb n_k}^{1/2}\mathbf{B}_{k,q}\mathbf{\Sigma}_{\mb n_k}^{1/2}|,\label{eq:firstIneq}
\end{align}
where the inequality is due to \eqref{lem:FI_Ineq}, and \eqref{eq:firstIneq} is due to \eqref{eq:covB}.
Let $\bar{\mathbf{B}}_k:= \sum_{q\in \mathcal{Q}}p(q)\mathbf{B}_{k,q}$. Then, we have from \eqref{eq:firstIneq}
\begin{align}
I(\mathbf{Y}_k;\mathbf{U}_k|\mathbf{X},Q) 
&\geq -  \sum_{q\in \mathcal{Q}}p(q) \log|\dv I-\mathbf{\Sigma}_{\mb n_k}^{1/2}\mathbf{B}_{k,q}\mathbf{\Sigma}_{\mb n_k}^{1/2}|\nonumber \\
&\geq-\log|\dv I-\mathbf{\Sigma}_{\mb n_k}^{1/2}\bar{\mathbf{B}}_k\mathbf{\Sigma}_{\mb n_k}^{1/2}|,\label{eq:logDetProp2}
\end{align}
where  \eqref{eq:logDetProp2} follows from the concavity of the log-det function and Jensen's inequality. 
On the other hand, we have
\begin{align}
I(\mathbf{X};\mathbf{U}_{S^c,q}|q)
&\leq \log|\mathbf{\Sigma}_{\mb x} |-\log|\mathbf{J}^{-1}(\mathbf{X}|\mathbf{U}_{S^c,q},q)|\label{eq:FI_Ineq},\\
& = \log 
\left| \sum_{k\in\mathcal{S}^{c}}\mathbf{\Sigma}_{\mb x}^{1/2}\mathbf{H}_{k}^{H}
\mathbf{B}_{k,q}
\mathbf{H}_{k}\mathbf{\Sigma}_{\mb x}^{1/2}\!+\!\mathbf{I}\right|\label{eq:secondtIneq},
\end{align}
where \eqref{eq:FI_Ineq} is due to \eqref{lem:FI_Ineq}; and \eqref{eq:secondtIneq} is due to to the following equality, which can be proved using the
connection between the MMSE matrix \eqref{eq:covB} and the Fisher information along the lines of 
\cite{Ekrem:IT:2014:OuterBoundCEO, DBLP:journals/corr/ZhouX0C16, EZSC:ISIT:2017} (We refer to \cite{Estella:IZS:2017:Proofs} for details):
\begin{align}
\mathbf{J}(\mathbf{X}|\mathbf{U}_{S^c,q},q) = \sum_{k\in\mathcal{S}^{c}}\mathbf{H}_{k}^{H}
\mathbf{B}_{k,q}
\mathbf{H}_{k}+\mathbf{\Sigma}_{\mb x}^{-1}\label{eq:Fischerequality}.
\end{align}

Similarly to \eqref{eq:logDetProp2}, from \eqref{eq:secondtIneq} and Jensen's Inequality we have
\begin{align}
I(\mathbf{X};\mathbf{U}_{S^c}|Q)
&\leq
\log\left| \sum_{k\in\mathcal{S}^{c}}\mathbf{\Sigma}_{\mb x}^{1/2}\mathbf{H}_{k}^{H}
\bar{\mathbf{B}}_{k}
\mathbf{H}_{k}\mathbf{\Sigma}_{\mb x}^{1/2}+\mathbf{I}\right|
\label{eq:secondtIneq_4}. 
\end{align}

Substituting \eqref{eq:logDetProp2} and \eqref{eq:secondtIneq_4} in \eqref{eq:MK_C_Main}
and letting $\mathbf{B}_k := \mathbf{\Sigma}_{\mb n_k}^{-1/2}\bar{\mathbf{B}}_k\mathbf{\Sigma}_{\mb n_k}^{-1/2}$ gives the desired outer bound.
The proof is completed by noting that the outer bound is achieved with $Q= \emptyset$ and $p^{*}(\dv u_k|\dv y_k,q) = \mc{CN}(\dv y_k, \dv \mathbf{\Sigma}_{\mb n_k}^{1/2}(\dv B_k-\dv I)\mathbf{\Sigma}_{\mb n_k}^{1/2}  ) $.

\begin{figure}[t!]
\centering
\includegraphics[width=0.475\textwidth]{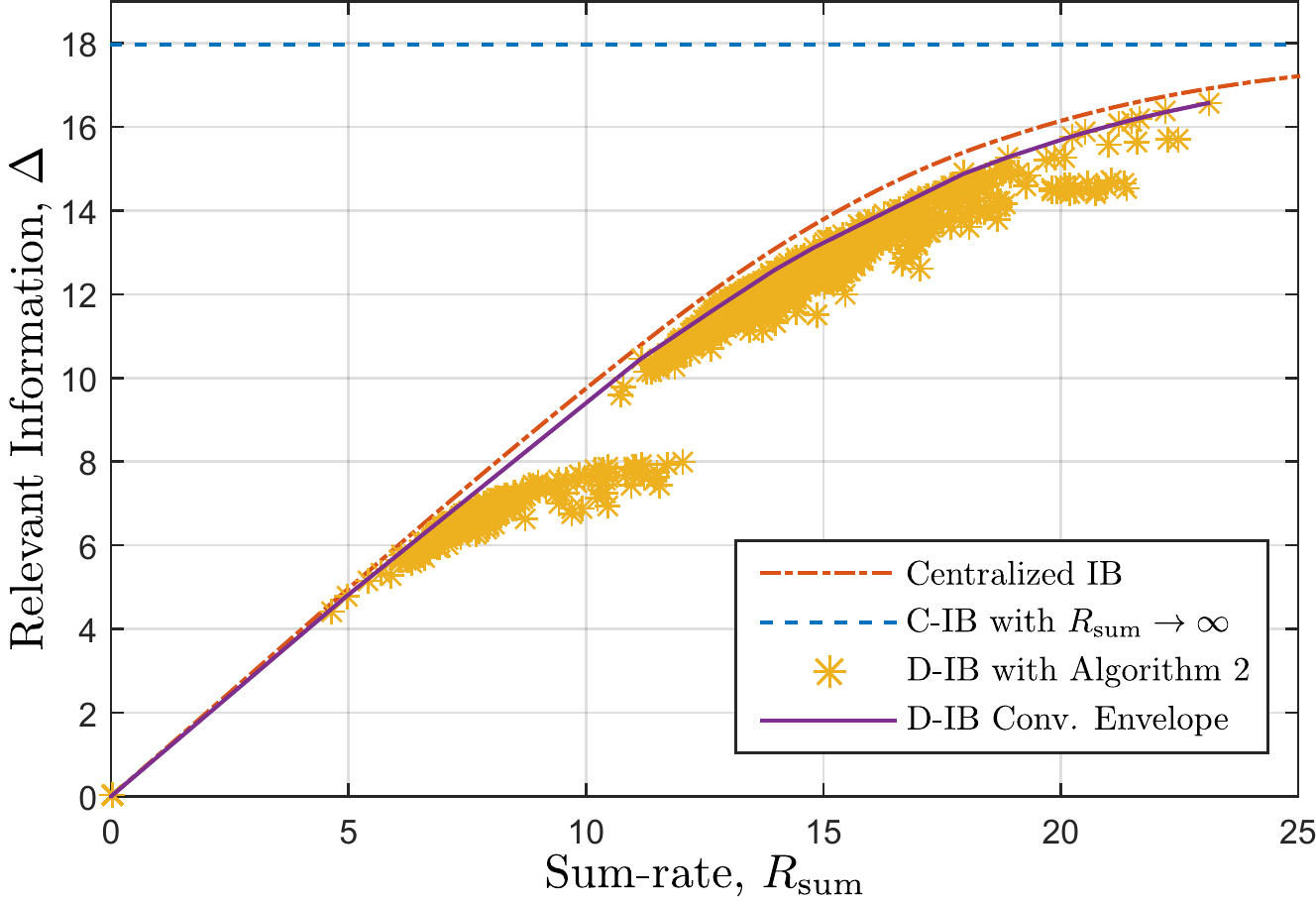}
\vspace{-4mm}
\caption{Information vs. sum-rate for  vector Gaussian D-IB with $K=2$ encoders, source dimension $N=4$, and observation dimension \mbox{$M_1 = M_2 = 2$.}} 
\label{fig:Figure1}
\vspace{-5mm}
\end{figure}

\vspace{-3mm}

\bibliographystyle{ieeetran}
\bibliography{ref}

\end{document}